\documentclass[11pt]{article}
\usepackage[latin1]{inputenc}
\usepackage{graphicx}
\usepackage{epsfig}
\usepackage{color}
\usepackage{vmargin}
\usepackage{multirow}
\usepackage{lscape}
\usepackage{verbatim}
\usepackage{amsthm, amssymb}
\usepackage{amsmath}
\newtheorem{Lem}{Lemma}
\newtheorem{theorem}{Theorem}

\def\dist{\operatorname{dist}}
\def\bdist{\operatorname{bdist}}
\def\refine{\operatorname{refine\_dist}}
\def\refinefurther{\operatorname{refine\_further}}
\def\even{\operatorname{even}}

\def\comb{\operatorname{comb}}

\title{Approximate Distance Oracles with Improved Query Time}
\author{Christian Wulff-Nilsen
        \footnote{Department of Mathematics and Computer Science,
                  University of Southern Denmark,
                  \texttt{koolooz@diku.dk},
                  \texttt{http://www.imada.sdu.dk/$_{\widetilde{~}}$cwn/}.}}

\date{}
\begin{document}

\maketitle
\begin{abstract}
Given an undirected graph $G$ with $m$ edges, $n$ vertices, and non-negative edge weights, and given an integer $k\geq 2$, we show
that a $(2k-1)$-approximate distance oracle for $G$ of size $O(kn^{1 + 1/k})$ and with $O(\log k)$ query time can be constructed in
$O(\min\{kmn^{1/k},\sqrt km + kn^{1 + c/\sqrt k}\})$ time for some constant $c$.
This improves the $O(k)$ query time of Thorup and Zwick.
Furthermore, for any $0 < \epsilon \leq 1$, we give an oracle of size $O(kn^{1 + 1/k})$ that
answers $((2 + \epsilon)k)$-approximate distance queries in $O(1/\epsilon)$ time.
At the cost of a $k$-factor in size, this improves the $128k$ approximation achieved by the constant query time oracle of
Mendel and Naor and approaches the best possible tradeoff between size and stretch, implied by a widely believed girth conjecture
of Erd\H{o}s. We can match the $O(n^{1 + 1/k})$ size bound of Mendel and Naor for any constant $\epsilon > 0$ and
$k = O(\log n/\log\log n)$.
\end{abstract}
\newpage

\section{Introduction}\label{sec:Intro}
The practical need for efficient algorithms to answer shortest path (distance) queries in graphs has increased significantly over
the years, in large part due to emerging GPS navigation technology and other route planning software.
Classical algorithms like Dijkstra do not scale well as they may need to explore the entire
graph just to answer a single query. As road maps are typically of considerable size, developing more efficient algorithms and data
structures has received a great deal of attention from the research community.

A distance oracle is a data structure that answers shortest path distance queries between vertex pairs in time independent of the size of the graph. A naive way
of achieving this is to precompute and store all-pairs shortest path distances in a look-up table, allowing subsequent
queries to be answered in
constant time. The obvious drawback is of course the huge space requirement which is quadratic in the number of vertices of the graph,
as well as the long time for precomputing all-pairs shortest path distances.

It is not difficult to see that quadratic space is necessary for constant
query time. It is therefore natural to consider \emph{approximate} distance
oracles where some error in the reported distances is allowed. We say that
an approximate distance $\tilde{d}_G(u,v)$ between two vertices $u$ and $v$
in a graph $G$ is of \emph{stretch} $\delta\geq 1$ if $d_G(u,v)\leq\tilde{d}_G(u,v)\leq\delta d_G(u,v)$, where $d_G(u,v)$ denotes the
shortest path distance in $G$ between $u$ and $v$. Awerbuch et
al.~\cite{Awerbuch98} gave for any integer $k\geq 1$ and a graph with $m$ edges and $n$ vertices
a data structure with stretch $64k$, space
$\tilde{O}(kn^{1 + 1/k})$, and preprocessing $\tilde{O}(mn^{1/k})$. Its query
time is $\tilde{O}(kn^{1/k})$ and therefore not independent of the size of
the graph. Stretch was improved to $2k + \epsilon$ by Cohen~\cite{Cohen98} and further to $2k - 1$ by Matou\v{s}ek~\cite{Matousek96}.

In the seminal paper of Thorup and Zwick~\cite{ThorupZwick}, it was shown
that a data structure of size $O(kn^{1 + 1/k})$ can be constructed in
$O(kmn^{1/k})$ time which reports shortest path distances stretched by a factor of at most $2k - 1$ in $O(k)$ time. Since its query time is independent of
the size of the graph (when $k$ is), we refer to it as an approximate distance
oracle. The tradeoff between size and stretch is optimal up to a factor of $k$ in space, assuming a widely believed
and partially proved girth conjecture of Erd\H{o}s~\cite{Erdos}.

Time and space in~\cite{ThorupZwick} are expected bounds; Roditty Thorup, and Zwick~\cite{DetOracleSpanner} gave a deterministic
oracle with only a small increase in preprocessing.

Baswana and Kavitha~\cite{APASP} showed how to obtain $O(n^2)$ preprocessing for $k\geq 3$, an improvement for dense graphs.
Subquadratic time was recently obtained for $k\geq 6$ and $m = o(n^2)$~\cite{SubquadraticOracleCWN}. P\u{a}tra\c{s}cu and
Roditty~\cite{Patrascu} gave an oracle of size $O(n^2/\alpha^{1/3})$ and stretch $2$ for a graph with $m = n^2/\alpha$ edges.
Furthermore, they showed that a size $O(n^{5/3})$ oracle with multiplicative stretch $2$ and additive stretch $1$ exists for unweighted
graphs. Baswana, Gaur, Sen, and Upadhyay~\cite{SubquadraticOracle} also gave oracles with both multiplicative and additive stretch.

Although the oracles above answer queries in time independent of the graph size, query time still depends on
stretch. Mendel and Naor~\cite{MendelNaor} asked the question of whether good approximate distance oracles exist with query time
bounded by a universal constant. They answered this in the affirmative by giving an oracle of size $O(n^{1 + 1/k})$, stretch at
most $128k$, query time $O(1)$ and preprocessing time $O(n^{2 + 1/k}\log n)$. Combining results of Naor and Tao~\cite{NaorTao} with Mendel and 
Naor~\cite{MendelNaor}
improves stretch to roughly $33k$; according to Naor and Tao, with a more
careful analysis of the arguments in~\cite{MendelNaor}, it should be possible
to further improve stretch to roughly $16k$ but not by much more.
The $O(n^{2 + 1/k}\log n)$ preprocessing time was later improved by Mendel and Schwob~\cite{CKR} to $O(mn^{1/k}\log^3n)$; for an
$n$-point metric space, they obtain a bound of $O(n^2)$.\footnote{I thank an anonymous referee for mentioning this improvement.}

We refer the reader to the survey by Sen~\cite{DistOracleSpannerSurvey} on distance oracles as well as the related area of spanners.

\paragraph{Our contributions:}
Our first contribution is an improvement of the query time of the Thorup-Zwick oracle from $O(k)$ to $O(\log k)$ without
increasing space, stretch, or preprocessing time. We achieve this by showing how to apply binary search on the bunch-structures,
introduced by Thorup and Zwick. Our improved query algorithm is very simple to describe and straightforward to implement.
It can easily be incorporated into our recent distance oracle~\cite{SubquadraticOracleCWN}, giving improved
preprocessing.

Our second contribution is an approximate distance oracle with universally constant query time whose size is $O(kn^{1 + 1/k})$
and whose stretch can be made arbitrarily close to the optimal $2k - 1$ (when $k = \omega(1)$): for any positive
$\epsilon\leq 1$, we give an oracle of size $O(kn^{1 + 1/k})$, stretch $O((2 + \epsilon)k)$,
and query time $O(1/\epsilon)$. For $k = O(\log n/\log\log n)$ and constant $\epsilon$, space can be improved to
$O(n^{1 + 1/k})$, matching that of Mendel and Naor\footnote{This covers almost all values of $k$ that are of interest
as the Mendel-Naor oracle has $O(n)$ space requirement for $k = \Omega(\log n)$.}.
To achieve this result, the main idea is to first query the Mendel-Naor oracle to get an $O(k)$-approximate distance
and then refine this estimate in $O(1/\epsilon)$ iterations using the
bunch-structures of Thorup and Zwick. Our results are summarized in Table~\ref{tab:Oracles}.

Note that we are interested in non-constant $k$ only; if $k = O(1)$, the Thorup-Zwick oracle is optimal up to constants
(assuming the girth conjecture) since it has size $O(n^{1 + 1/k})$, stretch $2k - 1$, and query time $O(1)$.

\begin{table*}
\begin{center}
\begin{tabular}{|c|c|c|c|c|}
\hline
Stretch        & Query time & Space      & Preprocessing time      & Reference\\
\hline
$2k-1$ & $O(k)$ & $O(kn^{1 + \frac 1 k})$ & $O(\min\{kmn^{\frac 1 k},\sqrt km + kn^{1 + \frac c{\sqrt k}}\})$ & \cite{ThorupZwick, SubquadraticOracleCWN}\\
\hline
$2k-1$ & $O(\log k)$ & $O(kn^{1 + \frac 1 k})$ & $O(\min\{kmn^{\frac 1 k},\sqrt km + kn^{1 + \frac c{\sqrt k}}\})$ & This paper\\
\hline
$128k$ & $O(1)$ & $O(n^{1 + \frac 1 k})$ & $O(mn^{\frac 1 k}\log^3n)$ & ~\cite{MendelNaor,CKR}\\
\hline
$(2 + \epsilon)k$ & $O(\frac {\log C}\epsilon)$ & $O(kn^{1 + \frac 1 k})$ &
$O(kmn^{\frac 1 k} + kn^{1+\frac 1k}\log n + mn^{\frac 1{Ck}}\log^3n)$ & This paper\\
\hline
\end{tabular}
\end{center}
\caption{Performance of distance oracles in weighted undirected graphs.}\label{tab:Oracles}
\end{table*}

\paragraph{Organization of the paper:}
In Section~\ref{sec:Prelim}, we introduce notation and give some basic definitions and results.
Our oracle with $O(\log k)$ query time is presented in Section~\ref{sec:FasterQuery}. This is followed by
our constant time oracle in Section~\ref{sec:ConstQuery}; first we present a generic algorithm in
Section~\ref{subsec:Generic} that takes as input a large-stretch distance estimate and outputs a refined
estimate. Some technical results are presented in Section~\ref{subsec:MendelNaor} that will allow us to combine
this generic algorithm with the Mendel-Naor oracle to form our own oracle. We describe preprocessing and query in detail
in Sections~\ref{subsec:Preproc} and~\ref{subsec:Query} and we bound time and space requirements in
Section~\ref{subsec:TimeSpace}. In Section~\ref{subsec:FasterPreproc}, we show how to improve preprocessing compared
to that in~\cite{CKR}.
Finally, we conclude in Section~\ref{sec:ConclRem}.

\section{Preliminaries}\label{sec:Prelim}
Throughout the paper, $G = (V,E)$ is an undirected connected graph with non-negative edge weights and with $m$ edges and $n$
vertices. For $u,v\in V$, we denote by $d_G(u,v)$ the shortest path distance between $u$ and $v$.

Sometimes we consider list representations of sets. We denote by $S[i]$ the $i$th entry of some chosen
list representation of a set $S$, $i\geq 0$. For $x > 0$, $\log x$ is the base $2$ logarithm of $x$.

The following definitions are taken from~\cite{ThorupZwick} and we shall use them throughout the paper.
Let $k\geq 1$ be an integer and form sets $A_0,\ldots,A_k$ with
$V = A_0\supseteq A_1\supseteq A_2\ldots\supseteq A_k = \emptyset$. For $i = 1,\ldots,k-1$, set $A_i$ is formed by
picking each element of $A_{i-1}$ independently with probability $n^{-1/k}$. Set $A_i$ has expected size $O(n^{1 - i/k})$ for
$i = 0,\ldots,k-1$. For each vertex $u$ and each $i = 1,\ldots,k-1$,
$p_i(u)$ denotes the vertex of $A_i$ closest to $u$ (breaking ties arbitrarily). Define a \emph{bunch} $B_u$ as
\[
  B_u = \bigcup_{i = 0}^{k - 1}\{v\in A_i\setminus A_{i+1} | d_G(u,v) < d_G(u,p_{i+1}(u))\},
\]
where we let $d_G(u,p_k(u)) = \infty$; see Figure~\ref{fig:Bunch}.
\begin{figure}
\centerline{\scalebox{1.2}{\input{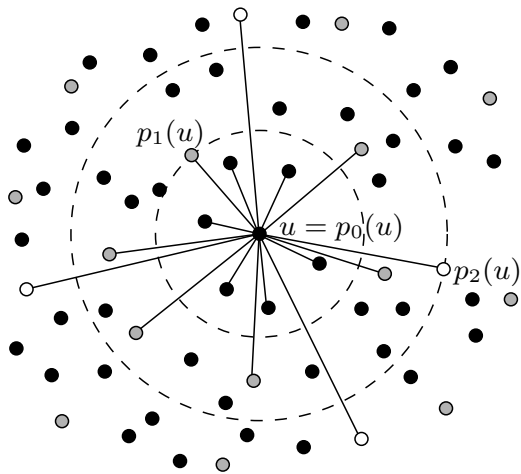}}}
\caption{A bunch $B_u$ in a complete Euclidean graph with $k = 3$. Black vertices belong to
$A_0$, grey vertices to $A_1$, and white vertices to $A_2$. Line segments connect $u$ to
vertices of $B_u$.}
\label{fig:Bunch}
\end{figure}

Thorup and Zwick showed how to compute all bunches in $O(kmn^{1/k})$ time and showed that each of them has expected size
$O(kn^{1/k})$ for a total of $O(kn^{1 + 1/k})$. The following lemma states some simple but important results about bunches.
\begin{Lem}\label{Lem:SimpleProperties}
Let $u,v\in V$ be distinct vertices and let $0\leq i < k - 1$. If $p_i(v)\notin B_u$ then $d_G(u,p_{i+1}(u))\leq d_G(u,p_i(v))$.
Furthermore, $A_{k - 1}\subset B_u$. In particular, $p_{k - 1}(v)\in B_u$.
\end{Lem}
Algorithm $\dist_k(u,v,i)$ in Figure~\ref{fig:Distk} is identical to the query algorithm of Thorup and Zwick except that we do not
initialize $i\leftarrow 0$ but allow any start value. We shall use this generalized algorithm in our analysis in the following.
\begin{figure}
\begin{tabbing}
d\=dd\=\quad\=\quad\=\quad\=\quad\=\quad\=\quad\=\quad\=\quad\=\quad\=\quad\=\quad\=\kill
\>\texttt{Algorithm} $\dist_k(u,v,i)$\\\\
\>1. \>$w\leftarrow p_i(u)$; $j\leftarrow i$\\
\>2. \>while $w\notin B_v$\\
\>3. \>\>$j\leftarrow j + 1$\\
\>4. \>\>$(u,v)\leftarrow (v,u)$\\
\>5. \>\>$w\leftarrow p_j(u)$\\
\>6. \>return $d_G(w,u) + d_G(w,v)$
\end{tabbing}
\caption{Answering a distance query, starting at sample level $i$.}\label{fig:Distk}
\end{figure}


\section{Oracle with $O(\log k)$ Query Time}\label{sec:FasterQuery}
In this section, we show how to improve the $O(k)$ query time of the Thorup-Zwick oracle to $O(\log k)$. Let $\mathcal I$ be the
index sequence $0,\ldots,k-1$. The idea is to identify $r = O(\log k)$ subsequences
$(\mathcal I_1 = \mathcal I)\supset\mathcal I_2\supset\ldots\supset\mathcal I_r$ of $\mathcal I$ in that order, where for
$j = 2,\ldots,r$,
$|\mathcal I_j|\leq\frac 1 2 |\mathcal I_{j-1}|$. Each subsequence $\mathcal I_j$ has the property that $\dist_k$ applied to the
beginning of it outputs a desired $(2k-1)$-approximate distance in $O(|\mathcal I_j|)$ time.
We apply binary search to identify the subsequences, with each step taking constant time. The final
subsequence $\mathcal I_r$ has $O(\log k)$ length and $\dist_k$ is applied to it to compute a $(2k-1)$-distance estimate in $O(\log k)$
additional time.

In the following, we define a class of such subsequences.
For vertices $u$ and $v$, an index $j\in\mathcal I$ is \emph{$(u,v)$-terminal} if
\begin{enumerate}
\item $j = k - 1$ (in which case $p_j(u)\in B_v$) or
\item $j < k - 1$ is even and either $p_j(u)\in B_v$ or $p_{j+1}(v)\in B_u$.
\end{enumerate}
Note that if an index $j$ is $(u,v)$-terminal, $\dist_k(u,v,i)$ terminates if it reaches $j$ or $j + 1$.
We say that a subsequence $\mathcal I' = i_1,\ldots,i_2$ of $\mathcal I$ is \emph{$(u,v)$-feasible} if
\begin{enumerate}
\item $i_1$ is even,
\item $d_G(u,p_{i_1}(u))\leq i_1\cdot d_G(u,v)$, and
\item $i_2$ is $(u,v)$-terminal.
\end{enumerate}
The following lemma implies that $\dist_k$ answers a $(2k - 1)$-approximate distance query for $u$ and $v$
when applied to a $(u,v)$-feasible sequence.
\begin{Lem}\label{Lem:FeasibleApprox}
Let $i_1,\ldots,i_2$ be a $(u,v)$-feasible subsequence. Then $\dist_k(u,v,i_1)$ gives a $(2k-1)$-approximate $uv$-distance
in $O(i_2 - i_1)$ time.
\end{Lem}
\begin{proof}
The time bound follows since $i_2$ is $(u,v)$-terminal
and since each iteration can be implemented to run in constant time using
hash tables to represent bunches as in~\cite{ThorupZwick}.
The stretch bound follows from the analysis of Thorup and Zwick for their query algorithm: when $p_j(u) = w\notin B_v$, we have
$d_G(v,p_{j+1}(v))\leq d_G(v,p_j(u))\leq d_G(u,p_j(u)) + d_G(u,v)$ by
Lemma~\ref{Lem:SimpleProperties} and the triangle inequality. Hence,
each iteration of $\dist_k(u,v,i_1)$ increases $d_G(w,u)$ by at most $d_G(u,v)$. Since $d_G(u,p_{i_1}(u))\leq i_1\cdot d_G(u,v)$, we have at termination
that $d_G(u,w) + d_G(w,v)\leq 2d_G(u,w) + d_G(u,v)\leq
(2(i_1 + (i_2 - i_1)) + 1)d_G(u,v)\leq (2k - 1)d_G(u,v)$.
\end{proof}
\begin{Lem}\label{Lem:IFeasible}
$\mathcal I$ is $(u,v)$-feasible for all vertices $u$ and $v$.
\end{Lem}
For each vertex $u$ and $0 \leq i < k - 2$, define $\delta_i(u) = d_G(u,p_{i+2}(u)) - d_G(u,p_i(u))$.
The following lemma allows us to binary search for a $(2k-1)$-approximate distance estimate of $d_G(u,v)$.
\begin{Lem}\label{Lem:FeasibleRec}
Let $i_1,\ldots,i_2$ be a $(u,v)$-feasible sequence and let $i$ be even, $i_1 + 2\leq i\leq i_2 - 2$.
Let $j$ be an even index in subsequence
$i_1,\ldots,i - 2$ that maximizes $\delta_j(u)$. If $p_j(u)\notin B_v$ and $p_{j+1}(v)\notin B_u$ then $i,\ldots,i_2$ is
$(u,v)$-feasible. Otherwise, $i_1,\ldots,j$ is $(u,v)$-feasible.
\end{Lem}
\begin{proof}
If $p_j(u)\in B_v$ or $p_{j+1}(v)\in B_u$ then $j$ is $(u,v)$-terminal. Since $i_1,\ldots,i_2$ is $(u,v)$-feasible, so is
$i_1,\ldots,j$.

Now assume that $p_j(u)\notin B_v$ and $p_{j+1}(v)\notin B_u$. Then $d_G(v,p_{j+1}(v))\leq d_G(v,p_j(u))$ and
$d_G(u,p_{j+2}(u))\leq d_G(u,p_{j+1}(v))$ by Lemma~\ref{Lem:SimpleProperties}.
Applying the triangle inequality twice yields
\begin{align*}
  d_G(u,p_{j+2}(u)) & \leq d_G(u,p_{j+1}(v))\\
                   & \leq d_G(u,v) + d_G(v,p_{j+1}(v))\\
                   & \leq d_G(u,v) + d_G(v,p_j(u))\\
                   & \leq 2d_G(u,v) + d_G(u,p_j(u))
\end{align*}
so $\delta_j(u) = d_G(u,p_{j+2}(u)) - d_G(u,p_j(u))\leq 2d_G(u,v)$.

Let $\mathcal I'$ be the set of even indices $i_1,i_1 + 2, i_1 + 4,\ldots,i - 2$.
Since $i_1,\ldots,i_2$ is $(u,v)$-feasible, $d_G(u,p_{i_1}(u))\leq i_1\cdot d_G(u,v)$. By the choice of $j$,
\begin{align*}
  d_G(u,p_i(u)) &    = d_G(u,p_{i_1}(u)) + \sum_{j'\in\mathcal I'}\delta_{j'}(u)\\
               &  \leq i_1\cdot d_G(u,v) + |\mathcal I'|\max_{j'\in\mathcal I'}\delta_{j'}(u)\\
                &    = i_1\cdot d_G(u,v) + \frac{i - i_1}{2}\delta_j(u)\\
                & \leq i_1\cdot d_G(u,v) + (i - i_1)d_G(u,v)\\
                &    = i\cdot d_G(u,v).
\end{align*}
Hence, since $i_1,\ldots,i_2$ is $(u,v)$-feasible, so is $i,\ldots,i_2$.
\end{proof}
We can now show our first main result.
\begin{figure}
\begin{tabbing}
d\=dd\=\quad\=\quad\=\quad\=\quad\=\quad\=\quad\=\quad\=\quad\=\quad\=\quad\=\quad\=\kill
\>\texttt{Algorithm} $\bdist_k(u,v,i_1,i_2)$\\\\
\>1. \>if $i_2 - i_1\leq\log k$ then return $\dist_k(u,v,i_1)$\\
\>2. \>let $i$ be the middle even index in $i_1,\ldots,i_2$\\
\>3. \>let $j$ be the (precomputed) even index in
       $i_1,\ldots,i-2$ maximizing $\delta_j(u)$\\
\>4. \>if $p_j(u)\notin B_v$ and $p_{j+1}(v)\notin B_u$ then
       return $\bdist_k(u,v,i,i_2)$\\
\>5. \>else return $\bdist_k(u,v,i_1,j)$
\end{tabbing}
\caption{Answering a distance query using binary search. The initial call
is $\bdist_k(u,v,0,k-1)$. For correctness of the pseudocode, we assume here that $k\geq 16$. The call in line $1$ is to $\dist_k$ in Figure~\ref{fig:Distk}.}\label{fig:BDistk}
\end{figure}
\begin{theorem}\label{Thm:LogkQuery}
For an integer $k\geq 2$, a $(2k-1)$-approximate
distance oracle of $G$ of size $O(kn^{1 + 1/k})$ and $O(\log k)$ query time can be constructed in
$O(\min\{kmn^{1/k},\sqrt km + kn^{1 + c/\sqrt k}\})$ time for some constant $c$.
\end{theorem}
\begin{proof}
In order for $\delta_i(u)$-values to be defined, we assume that
$k\geq 3$; the result of the theorem is already known for $k = 2$ (in fact
for any constant $k$).
We obtain bunch $B_u$ for each vertex $u$ in a total of $O(kmn^{1/k})$ time using the Thorup-Zwick construction.
The following additional preprocessing is done for $u$ to determine the $(u,v)$-subsequences of $\mathcal I$ that are needed. Let
$\mathcal I' = i_1,\ldots,i_2$ be the current sequence considered; initially, $\mathcal I' = \mathcal I$. Pick an even index
$i$, $i_1 + 2\leq i\leq i_2 - 2$, such that $i_1,\ldots,i$ and $i,\ldots,i_2$ have (roughly) the same size and
find an even index $j$ in $i_1,\ldots,i - 2$ which maximizes $\delta_j(u)$. Then recurse on subsequences
$i_1,\ldots,j$ and $i,\ldots,i_2$. The recursion stops when a sequence of length at most $\log k$  is reached.
Below we show that these indices $j$ can be identified in $O(k)$ time which is $O(kn)$ over all $u$.

Now, to answer a distance query for vertices $u$ and $v$, we do binary search on sequences
$\mathcal I' = i_1,\ldots,i_2$ generated; see Figure~\ref{fig:BDistk}. We start the search with $\mathcal I' = \mathcal I$ and check if
both $p_j(u)\notin B_v$ and $p_{j+1}(v)\notin B_u$. If so, we continue the search on subsequence $i,\ldots,i_2$.
Otherwise, we continue the search on $i_1,\ldots,j$. We stop when reaching a sequence of length at most $\log k$.
By Lemmas~\ref{Lem:IFeasible} and~\ref{Lem:FeasibleRec}, this subsequence is $(u,v)$-feasible.
Applying $\dist_k$ to it outputs a $(2k-1)$-approximate distance estimate of $d_G(u,v)$ by
Lemma~\ref{Lem:FeasibleApprox}.

Binary search takes $O(\log k)$ time. Since we end up with a $(u,v)$-feasible sequence of length at most $\log k$, $\dist_k$
applied to it takes $O(\log k)$ time. Hence, query time is $O(\log k)$.

The oracle in~\cite{SubquadraticOracleCWN} with $O(\sqrt km + kn^{1 + c/\sqrt k})$ preprocessing time
also constructs bunches and applies linear search in these to answer distance
queries in $O(k)$ time. Our binary search algorithm can immediately be plugged in instead.

It remains to bound, for each vertex $u$, the time to identify the indices $j$. Since sequence lengths are reduced by a factor of at least two in each
recursive step, simple linear searches will give all the indices in a total of $O(k\log k)$ time. In the following, we improve this to $O(k)$.

Let us call a subsequence of $\mathcal I$ \emph{canonical} if it is obtained during the following procedure: start with the subsequence
$\mathcal I'$ of $\mathcal I$ consisting of the even indices. Then find an index $i\in\mathcal I'$ that partitions $\mathcal I'$ into two
(roughly) equal-size subsequences (both containing $i$), and recurse on each of
them; the recursion stops when a subsequence consisting of two indices is obtained. We keep a binary tree $\mathcal T$ reflecting the recursion,
where each node of $\mathcal T$ is associated with the canonical subsequence generated at that step in the recursion.
From this procedure, we identify (the endpoints of) all canonical subsequences in $O(k)$ time.
A bottom-up $O(k)$ time algorithm in $\mathcal T$ can then identify, for each canonical subsequence $\mathcal I' = i_1,i_1+2\ldots,i_2$, an
index $j = j(\mathcal I')$ in $i_1,i_1+2,\ldots,i_2 - 2$ that maximizes $\delta_j(u)$.

Now consider a (not necessarily canonical) subsequence $\mathcal I' = i_1,i_1 + 2,\ldots,i_2$ of $\mathcal I$ with indices $i_1 < i_2$ even. We can
find $O(\log k)$ canonical subsequences whose union is $\mathcal I'$ as follows: let $\ell_1$ and $\ell_2$ be the leaves
of $\mathcal T$ associated with canonical subsequences $i_1,i_1+2$ and $i_2 - 2,i_2$, respectively. Let $P$ be the path in $\mathcal T$ from the
parent of $\ell_1$ to the parent of $\ell_2$ and let $X$ be the set of nodes in $\mathcal T\setminus P$ having a parent in $P$. Then it is
easy to see that the $O(\log k)$ canonical subsequences associated with nodes in $X$ have $\mathcal I'$ as their union.
It follows that finding the desired index $j$ for $\mathcal I'$ takes $O(\log k)$ time as it can be found among the $j$-indices for canonical
subsequences associated with nodes in $X$.

In our preprocessing for vertex $u$, we only need to find $j$-indices for $O(k/\log k)$ subsequences since
the recursion stops when a subsequence of length at most $\log k$ is found.
Total preprocessing for $u$ is thus $O(k)$, as desired.
\end{proof}

\section{Oracle with Constant Query Time}\label{sec:ConstQuery}
Let $0 < \epsilon \leq \frac 1 2$ be given. In this section, we show how to achieve stretch $2(1 + \epsilon)k - 1$, query time
$O(1/\log(1 + \epsilon)) = O(1/\epsilon)$\footnote{Let $x = 1/\epsilon\geq 2$.
Since $\ln$ is concave, $\ln(1 + \epsilon) = \ln(x + 1) - \ln x > \frac{\partial}{\partial x}\ln(x + 1) = 1/(x + 1)\geq
\frac 2 3\epsilon$, which implies
$1/\log(1 + \epsilon) = O(1/\epsilon)$.},
and space $O(kn^{1 + 1/k})$. Initially, we aim for a preprocessing bound of $O(n^{2 + 1/k}\log n)$, matching that
in~\cite{MendelNaor}. In Section~\ref{subsec:FasterPreproc}, we improve this to the bound stated in
Table~\ref{tab:Oracles}.

We start with a generic algorithm, $\refine$, to refine a distance
estimate. Later we will show how to combine this with the Mendel-Naor oracle.
We shall assume that $1/\log(1 + \epsilon) = o(\log k)$ since otherwise, the oracle of the previous section can be applied.

\subsection{A generic algorithm}\label{subsec:Generic}
For a vertex $u$ and a non-negative value $d_u$, we define $\even_u(d_u)$ as the largest even index
$i_u$ such that $d_G(u,p_{i_u}(u))\leq d_u$.

Pseudocode of $\refine$ can be found in Figure~\ref{fig:refine}.
\begin{figure}
\begin{tabbing}
d\=ddd\=\quad\=\quad\=\quad\=\quad\=\quad\=\quad\=\quad\=\quad\=\quad\=\quad\=\quad\=\kill
\>\texttt{Algorithm} $\refine_{\alpha,\epsilon}(u,v,\tilde d_{uv})$\\\\
\>1. \>$d_u\leftarrow \tilde d_{uv}$\\
\>2. \>$i_u\leftarrow\even_u(d_u)$\\
\>3. \>if not $\refinefurther(u,v,i_u)$ then return $d_u$\\
\>4. \>$i\leftarrow 0$\\
\>5. \>while $\refinefurther(u,v,i_u)$ and
       $i\leq\lceil\log(2\alpha)/\log(1 + \epsilon)\rceil$\\
\>6. \>\>$d_u\leftarrow d_u/(1 + \epsilon)$\\
\>7. \>\>$i_u\leftarrow\even_u(d_u)$\\
\>8. \>\>$i\leftarrow i + 1$\\
\>9. \>$i_u'\leftarrow \even_u(d_u(1 + \epsilon))$\\
\>10.\>if $i_u'\geq 2$ then\\
\>11.\>\>let $j$ be an even index in $0,\ldots,i_u' - 2$
         maximizing $\delta_j$\\
\>12.\>\>if $p_j(u)\in B_v$ then
         return $d_G(u,p_j(u)) + d_G(v,p_j(u))$\\
\>13.\>\>if $p_{j+1}(v)\in B_u$ then
         return $d_G(u,p_{j+1}(v)) + d_G(v,p_{j+1}(v))$\\
\>14.\>if $p_{i_u'}(u)\in B_v$ then
       return $d_G(u,p_{i_u'}(u)) + d_G(v,p_{i_u'}(u))$\\
\>15.\>else return $d_G(u,p_{i_u'+1}(v)) + d_G(v,p_{i_u'+1}(v))$\\
\end{tabbing}
\begin{tabbing}
d\=ddd\=\quad\=\quad\=\quad\=\quad\=\quad\=\quad\=\quad\=\quad\=\quad\=\quad\=\quad\=\kill
\>\texttt{Algorithm} $\refinefurther(u,v,i_u)$\\\\
\>1. \>if $i_u \geq 2$ then\\
\>2. \>\>let $j$ be an even index in $0,\ldots,i_u - 2$
         maximizing $\delta_j$\\
\>3. \>\>if $p_j(u)\in B_v$ or $p_{j+1}(v)\in B_u$ then
         return \texttt{true}\\
\>4. \>if $p_{i_u}(u)\in B_v$ or $p_{i_u + 1}(v)\in B_u$ then
       return \texttt{true}\\
\>5. \>else return \texttt{false}\\
\end{tabbing}
\caption{Algorithm $\refine$ takes as input an $\alpha k$-approximate $uv$-distance $\tilde d_{uv}$ and outputs
         a $(2(1 + \epsilon)k - 1)$-approximate $uv$-distance.}\label{fig:refine}
\end{figure}
It takes as input an $\alpha k$-approximate $uv$-distance $\tilde d_{uv}$ and outputs
a $(2(1 + \epsilon)k - 1)$-approximate $uv$-distance. In line $3$, it calls
subroutine $\refinefurther$ which checks a condition similar to that in Lemma~\ref{Lem:FeasibleRec}
to determine whether the initial estimate $\tilde d_{uv}$ is already a good enough approximation.
If so, $\refine$ outputs this distance in line $3$. Otherwise, it repeatedly refines the initial
estimate in the while loop in lines $5$--$8$. In each iteration, the estimate is reduced by a
factor of $(1 + \epsilon)$ and $\refinefurther$ is called to determine whether we can refine the
estimate further. If not, the while-loop ends and the refined estimate is output in lines
$12$--$15$. The while-loop also terminates after roughly $\log \alpha/\epsilon$ iterations since
then the refined estimate is small enough as the initial estimate is an $\alpha k$-approximate
$uv$-distance. With the Mendel-Naor oracle, we can pick $\alpha = 128$, giving only $O(1/\epsilon)$
iterations. We will implement $\refine$ so that each iteration takes $O(1)$ time, giving the
desired $O(1/\epsilon)$ query time.

The following lemma shows that $\refine$ outputs the stretch we are aiming for.
\begin{Lem}\label{Lem:Generic}
For $k\geq 4$, $\alpha\geq 1$, and $\epsilon > 0$, algorithm $\refine_{\alpha,\epsilon}(u,v,\tilde d_{uv})$ outputs a
$(2(1 + \epsilon)k - 1)$-approximate $uv$-distance if $\tilde d_{uv}$ is an $\alpha k$-approximate $uv$-distance.
\end{Lem}
\begin{proof}
Initially, $d_G(u,v)\leq\tilde d_{uv} = d_u$. If the test in line $3$ of
$\refine$ succeeds, i.e., if algorithm $\refinefurther$ returns \texttt{false}, then
since the test in line $3$ of that algorithm fails, a telescoping
sums argument similar to that in the proof of
Lemma~\ref{Lem:FeasibleRec} implies
$d_G(u,p_{i_u}(u))\leq i_u\cdot d_G(u,v)$. Since also line
$4$ fails, we have
$d_G(u, p_{i_u + 2}(u)) - d_G(u,p_{i_u}(u))\leq 2d_G(u,v)$.
Hence $d_G(u,v)\leq d_u < d_G(u,p_{i_u + 2}(u))\leq (i_u + 2)d_G(u,v) \leq (k - 1)d_G(u,v)$
(note that $i_u + 2\leq k - 1$ since $p_{i_u + 1}(v)\notin B_u$ which implies $i_u + 1 < k - 1$ by
Lemma~\ref{Lem:SimpleProperties}). In the following, we can thus assume
that the test in line $3$ of $\refine$ fails.

We know that $\refinefurther(u,v,i_u')$ returns \texttt{true} since $i_u'$ is the value of $i_u$ in the iteration before the last.
Hence, if a distance is returned in line $15$, $p_{i_u' + 1}(v)\in B_u$. In particular, all distances returned are at least
$d_G(u,v)$.

Assume first that the while-loop ended because $\refinefurther(u,v,i_u)$ returned \texttt{false}. Observing the following string
of inequalities in lines $10$ to $15$ will help us in the following:
\[
  d_G(u,p_{i_u}(u)) \leq d_u < d_G(u,p_{i_u + 2}(u))
                   \leq d_G(u,p_{i_u'}(u))
                   \leq d_u(1 + \epsilon).
\]
We have $d_u < d_G(u,p_{i_u + 2}(u))\leq (i_u + 2)d_G(u,v)$. If lines $11$ to $13$ are executed then
$d_G(u,p_j(u)) < d_G(u, p_{i_u'}(u))\leq d_u(1 + \epsilon) < (1 + \epsilon)(i_u + 2)d_G(u,v)$. Thus,
if $p_j(u)\in B_v$, a value of at most
\begin{align*}
2d_G(u,p_j(u)) + d_G(u,v) & <    (2(1 + \epsilon)(i_u + 2) + 1)d_G(u,v)\\
                         &  \leq (2(1 + \epsilon)(k - 1) + 1)d_G(u,v)\\
                          &    < (2(1 + \epsilon)k - 1)d_G(u,v)
\end{align*}
is returned in line $12$. If $p_j(u)\notin B_v$ and $p_{j+1}(v)\in B_u$, Lemma~\ref{Lem:SimpleProperties} gives $j + 1 \leq k - 1$
and
\[
  d_G(v,p_{j+1}(v)) \leq d_G(v,p_j(u))
                   \leq d_G(u,v) + d_G(u,p_j(u))
                   < ((1 + \epsilon)(i_u + 2) + 1)d_G(u,v).
\]
Furthermore, since $p_{j+1}(v)\in B_u$ and $j + 1 < i_u'$, we have
\[
  d_G(u,p_{j+1}(v)) \leq d_G(u,p_{j+2}(u))
                   \leq d_G(u,p_{i_u'}(u))
                   \leq d_u(1 + \epsilon)
                   < (1 + \epsilon)(i_u + 2)d_G(u,v).
\]
Hence, a value of less than
\[
(2(1 + \epsilon)(i_u + 2) + 1)d_G(u,v) \leq
(2(1 + \epsilon)(k - 1) + 1)d_G(u,v) <
(2(1 + \epsilon)k - 1)d_G(u,v)
\]
is returned in line $13$. The same argument as for line $12$
with $i_u'$ instead of $j$ shows that the desired distance estimate is output
in line $14$. If we reach line $15$, $p_{i_u'}(u)\notin B_v$ and (as already observed) $p_{i_u' + 1}(v)\in B_u$.
Then $i_u + 2\leq i_u' \leq k - 2$ and
\begin{align*}
  d_G(v,p_{i_u' + 1}(v)) & \leq d_G(v,p_{i_u'}(u))\\
                      & \leq d_G(u,v) + d_G(u,p_{i_u'}(u))\\
                      & \leq d_G(u,v) + d_u(1 + \epsilon)\\
                         &    < ((1 + \epsilon)(i_u + 2) + 1)d_G(u,v)\\
                         & \leq ((1 + \epsilon)(k - 2) + 1)d_G(u,v)\\
                         &    < ((1 + \epsilon)k - 1)d_G(u,v)
\end{align*}
so a value of at most
\[
  2d_G(v,p_{i_u'+1}(v)) + d_G(u,v) <
  (2((1 + \epsilon)k - 1) + 1)d_G(u,v) =
  (2(1 + \epsilon)k - 1)d_G(u,v)
\]
is returned in line $15$.

Now assume that the while-loop ended with $\refinefurther(u,v,i_u)$ returning \texttt{true}.
Then $i = \lceil\log(2\alpha)/\log(1 + \epsilon)\rceil$ iterations have been executed so
the final value of $d_u$ is at most
$\alpha k\cdot d_G(u,v)/(1 + \epsilon)^i\leq\frac k 2\cdot d_G(u,v)$. If the algorithm returns a value in line $12$ then
this value is at most
$2d_G(u,p_j(u)) + d_G(u,v) < 2d_u(1 + \epsilon) + d_G(u,v)\leq ((1 + \epsilon)k + 1)d_G(u,v)$.
If $p_j(u)\notin B_v$ and $p_{j+1}(v)\in B_u$ then
$d_G(v,p_{j+1}(v))\leq d_G(v,p_j(u))\leq d_G(u,v) + d_G(u,p_j(u)) < d_G(u,v) + d_u(1 + \epsilon)$ so a value of at most
$2d_G(v,p_{j+1}(v)) + d_G(u,v) < 2d_u(1 + \epsilon) + 3d_G(u,v) \leq ((1 + \epsilon)k + 3)d_G(u,v)$ is returned in line
$13$. Since $k\geq 4$, this gives the desired estimate. A similar argument gives the same estimate for lines $14$ and $15$.
This completes the proof.
\end{proof}

\subsection{Combining with the Mendel-Naor oracle}\label{subsec:MendelNaor}
Our oracle will query that of Mendel and Naor for a distance estimate and then give it as input to an efficient implementation of
$\refine$. It is worth pointing out that any oracle with universally constant
query time and $O(k)$ stretch can be used as a black box and not just that
in~\cite{MendelNaor}; the only requirement
is that the number of distinct distances it can output is not too big; see
details below.

We will keep a sorted list of values such that for any distance query, the list contains
the $O(1/\log(1 + \epsilon))$ $d_u$-values found in $\refine$ as consecutive entries. We linearly traverse the list to
identify these entries some of which point to $i_u$-indices needed by $\refine$. These pointers together with some
additional preprocessing allow us to execute each iteration of the while-loop in $O(1)$ time.

We will ensure the property that the elements of the list are spaced by a factor of at least $1 + \epsilon$. For this we need a new definition.
Let $S$ be a non-empty set of real numbers and let $\epsilon > 0$ be given. Define the \emph{$\epsilon$-comb} of $S$
to be the set $S_\epsilon$ of real numbers obtained by the iterative algorithm $\comb_\epsilon(S)$ in Figure~\ref{fig:comb}.
Lemmas~\ref{Lem:Comb} and~\ref{Lem:EpsilonComb} below show that the $\epsilon$-comb of a certain superset of the set of
all distances that can be output by the Mendel-Naor oracle has the above property while not containing too many elements.
\begin{figure}
\begin{tabbing}
d\=dd\=\quad\=\quad\=\quad\=\quad\=\quad\=\quad\=\quad\=\quad\=\quad\=\quad\=\quad\=\kill
\>\texttt{Algorithm} $\comb_\epsilon(S)$\\\\
\>1. \>let $s_{\max}$ be the largest element of $S$\\
\>2. \>$S_\epsilon \leftarrow \{s_{\max}\}$; $S'\leftarrow S\setminus\{s_{\max}\}$\\
\>3. \>while $S'\neq\emptyset$\\
\>4. \>\>let $s_1$ be the largest element of $S'$ and let $s_2$ be
         the smallest element of $S_\epsilon$\\
\>5. \>\>$s \leftarrow \min\{s_1,s_2/(1 + \epsilon)\}$\\
\>6. \>\>$S_\epsilon\leftarrow S_\epsilon\cup\{s\}$\\
\>7. \>\>remove all the elements from $S'$ that have value
         at least $s$\\
\>8. \>return $S_\epsilon$
\end{tabbing}
\caption{Algorithm that outputs the $\epsilon$-comb $S_\epsilon$ of a non-empty set $S$ of real values.}\label{fig:comb}
\end{figure}
\begin{Lem}\label{Lem:Comb}
Let $S_\epsilon$ be the $\epsilon$-comb of a set $S$. Then
\begin{enumerate}
\item for any $s\in S$, there is a unique $s'\in S_\epsilon$ such that $s\leq s' < (1 + \epsilon)s$,
\item any two elements of $S_\epsilon$ differ by a factor of at least $1 + \epsilon$, and
\item $|S_\epsilon|\leq |S|$.
\end{enumerate}
\end{Lem}
\begin{proof}
To show the first part,
define $s^{(i)}$ to be the element $s$ found in the $i$th iteration of the while-loop. Define $s_1^{(i)}$ and $s_2^{(i)}$
similarly. Now, let $s\in S$ be given. Since $s_{\max}\in S_\epsilon$, there is an element of $S_\epsilon$ which is at least $s$.
Let $s_{\min}$ be the smallest such element and suppose for the sake of contradiction that $s_{\min}\geq (1 + \epsilon)s$.
Let $i$ be the iteration in which $s_{\min}$ is added to $S_\epsilon$. Since $s < s^{(i)}$, $s = s^{(j)}$ for some $j\geq i + 1$
so $s\leq s_1^{(i+1)}$. After line $7$ has been executed, every element of
$S'$ is strictly smaller than $s^{(i)} = s_{\min}$. Thus, $s\leq s_1^{(i+1)} < s_{\min}$. Since also
$s_2^{(i+1)} = s^{(i)} = s_{\min}\geq (1 + \epsilon)s$, it follows that
$s\leq s^{(i+1)} < s_{\min}$. But $s^{(i+1)}\in S_\epsilon$, contradicting the choice of $s_{\min}$.

We have shown that $s\leq s_{\min}\leq (1 + \epsilon)s$. To show uniqueness, let $s'$ be the first element added to $S_{\epsilon}$
for which $s\leq s' < (1 + \epsilon)s$. Assume for the sake of contradiction that $s'\neq s_{\min}$. Then $s_{\min}$ was added
in a later iteration than $s'$ so $s\leq s_{\min} = s^{(i)}\leq s_2^{(i)}/(1 + \epsilon)\leq s'/(1 + \epsilon) < s$,
a contradiction. Thus, $s' = s_{\min}$, showing uniqueness.

The second part of the lemma holds since in line $5$, $s_2$ is the smallest element of $S_{\epsilon}$ and the
next element $s$ to be added to this set satisfies $s\leq s_2/(1 + \epsilon)$.

The third part of the lemma follows since in line $2$,
$|S_\epsilon| = 1$ and $S' = |S| - 1$ and since at least one element (namely $s_1$) is removed from $S'$ in line $7$ after
an element has been added to $S_\epsilon$.
\end{proof}

For any vertices $u$ and $v$, denote by $d_{MN}(u,v)$ the $uv$-distance estimate output by the Mendel-Naor oracle and let
$\alpha_{MN}k$ be the stretch achieved by the oracle, i.e., $\alpha_{MN} = 128$.
Let $\mathcal D_{MN} = \{d_{MN}(u,v) | u,v\in V\}$ be the set of all distances that the oracle can output.
\begin{Lem}\label{Lem:DMNBound}
$|\mathcal D_{MN}| = O(n^{1 + 1/k})$.
\end{Lem}
\begin{proof}
The Mendel-Naor oracle stores trees representing certain ultrametrics. Each tree node is labelled with a distance and each
approximate distance output
by the Mendel-Naor oracle is one such label. Hence, since the oracle has size $O(n^{1 + 1/k})$, so has $\mathcal D_{MN}$.
\end{proof}
\begin{Lem}\label{Lem:EpsilonComb}
For each $d\in\mathcal D_{MN}$, let
$\mathcal D_d = \{d/(1 + \epsilon)^i | 0\leq i\leq\lceil \log(2\alpha_{MN}(1 + \epsilon))/\log(1 + \epsilon)\rceil\}$ and let
$\mathcal D_\epsilon$ be the $\epsilon$-comb of $\cup_{d\in\mathcal D_{MN}}\mathcal D_d$. Then for each $d\in\mathcal D_{MN}$,
there exists a unique $d'\in\mathcal D_{\epsilon}$ such that $d\leq d'\leq d(1 + \epsilon)$ and
$d'/(1 + \epsilon)^i\in\mathcal D_{\epsilon}$ for $0\leq i\leq \lceil \log(2\alpha_{MN}(1 + \epsilon))/\log(1 + \epsilon)\rceil$.
Also, $|\mathcal D_\epsilon| = O(n^{1 + 1/k}/\log(1 + \epsilon))$.
\end{Lem}
\begin{proof}
The existence and uniqueness of $d'$ follows from $\mathcal D_{MN}\subset\bigcup_{d\in\mathcal D_{MN}}\mathcal D_d$ and from part
$1$ of Lemma~\ref{Lem:Comb}. Define $d_i = d/(1 + \epsilon)^i$ and $d_i' = d'/(1 + \epsilon)^i$.
We use induction on $i\geq 0$ to show that $d_i'\in\mathcal D_\epsilon$. The base case $i = 0$ has
been shown since $d_0' = d'$ so assume $0 < i\leq\lceil \log(2\alpha_{MN}(1 + \epsilon))/\log(1 + \epsilon)\rceil$
and that $d_{i-1}'\in\mathcal D_\epsilon$. Consider the iteration of
$\comb_\epsilon(\cup_{d\in\mathcal D_{MN}}\mathcal D_d)$ following that in which $d_{i-1}$ was added to $S_\epsilon$. Here,
$s_1\geq d_{i-1}$ since $d_{i-1}\in S'$ and so $s_2 = d_{i-1}' = d_i'(1 + \epsilon)\leq d_{i-1}(1 + \epsilon)\leq s_1(1 + \epsilon)$,
giving $s = \min\{s_1,s_2/(1 + \epsilon)\} = s_2/(1 + \epsilon) = d_i'$ which
is added to $S_{\epsilon}$ in line $6$. Hence, $d_i'\in\mathcal D_\epsilon$, completing the induction step.

For the last part of the lemma, since $\log(2\alpha_{MN}(1 + \epsilon))/\log(1 + \epsilon) = O(1/\log(1 + \epsilon))$,
Lemma~\ref{Lem:DMNBound} and part $3$ of Lemma~\ref{Lem:Comb} give
\[
  |\mathcal D_\epsilon| \leq\sum_{d\in\mathcal D_{MN}}|\mathcal D_d| = O(|\mathcal D_{MN}|/\log(1 + \epsilon)) = O(n^{1 + 1/k}/\log(1 + \epsilon)).
\]
\end{proof}

As mentioned earlier, certain elements of the $\epsilon$-comb in Lemma~\ref{Lem:EpsilonComb} contain pointers to $i_u$-indices.
These pointers are defined by the following type of map. For a set $S$ of real values with smallest element $s_{\min}$, define
$\tau_S:[s_{\min},\infty)\rightarrow S$ by $\tau_S(x) = \max\{s\in S | s\leq x\}$.
\begin{Lem}\label{Lem:tau}
Let $S$ be a set of real values with smallest element $s_{\min}$ and let  $x,y\in\mathcal [s_{\min},\infty)$.
If $s_1 < s_2$ are consecutive elements in $S$ then $\tau_S(x) = \tau_S(y) = s_1$ iff $x,y\in[s_1,s_2)$.
\end{Lem}

\subsection{Preprocessing}\label{subsec:Preproc}
We are now ready to give an efficient implementation of algorithm $\refine$. We construct the Mendel-Naor oracle
and obtain the set $\mathcal D_{MN}$. For each vertex $u$, we construct bunch $B_u$ and the set $P_u$ of
values $d_G(u,v)$ for each $v\in B_u$. We represent $P_u$ as a list sorted by increasing value. Furthermore,
we find a set $S_u$ of real values as follows. For each index $i\in\{0,\ldots, |P_u| - 2\}$ of $P_u$, subdivide interval
$[P_u[i],P_u[i+1]]$ into four even-length subintervals. We denote by $\mathcal I_u$ the set
of these subintervals over all $i$ and form the set $S_u$ of all their endpoints.
We obtain the $\epsilon$-comb $\mathcal D_\epsilon$ as defined in Lemma~\ref{Lem:EpsilonComb} and
represent it as a sorted list. Then we form a set $\mathcal D_\epsilon(u)$ of those $d\in\mathcal D_\epsilon$
for which $d$ is either the smallest or the largest element that $\tau_{S_u}$ maps to $\tau_{S_u}(d)$; see Figure~\ref{fig:Pointers}.
\begin{figure}
\centerline{\scalebox{1.0}{\input{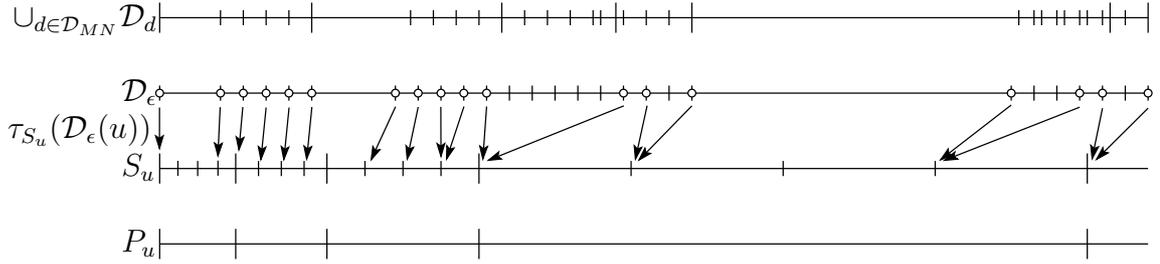}}}
\caption{Sets $\cup_{d\in\mathcal D_{MN}}\mathcal D_d$, $\mathcal D_\epsilon$, $S_u$, and $P_u$ (ordered by increasing value from left
         to right) as well as the map $\tau_{S_u}$ restricted to the subset $\mathcal D_\epsilon(u)$ (white points) of
         $\mathcal D_\epsilon$. Elements of $\cup_{d\in\mathcal D_{MN}}\mathcal D_d$ represented by long line segments are those
         belonging to $\mathcal D_{MN}$. For clarity, elements of each set $\mathcal D_d$ from Lemma~\ref{Lem:EpsilonComb}
         are evenly spaced in the figure.}
\label{fig:Pointers}
\end{figure}
With each $d\in\mathcal D_{\epsilon}(u)$, we associate the largest even index $i_u(d)$ such that
$d_G(u,p_{i_u(d)}(u))\leq\tau_{S_u}(d)$. For all $d\in\mathcal D_\epsilon\setminus\mathcal D_\epsilon(u)$, we leave
$i_u(d)$ undefined.

\subsection{Query}\label{subsec:Query}
To answer an approximate $uv$-distance query, we first obtain the Mendel-Naor estimate $d_{MN}(u,v)$ and
identify the smallest element $\tilde d_{uv}$ of $\mathcal D_{\epsilon}$ which is at least $d_{MN}(u,v)$.
This element is the input to $\refine_{\alpha,\epsilon}$ where $\alpha = (1 + \epsilon)\alpha_{MN}$.
By Lemma~\ref{Lem:EpsilonComb}, $\tilde d_{uv}$ is an $\alpha k$-approximate distance so the
output will be a $(2(1 + \epsilon)k - 1)$-approximate distance.

It follows from Lemma~\ref{Lem:EpsilonComb} and part $2$ of Lemma~\ref{Lem:Comb}
that all values of $d_u$ in $\refine$ are consecutive and start
from $\tilde d_{uv}$ in $\mathcal D_\epsilon$. Linearly traversing the list from $\tilde d_{uv}$ thus
corresponds to updating $d_u$ in the while-loop.

We also need to maintain even index $i_u$. Assume for now that for the initial $d_u$, index
$i_u(d_u)$ is defined. Then the initial $i_u$ is $i_u(d_u)$. As $d_u$ is updated in the while-loop,
at some point it may happen that $i_u(d_u)$ is undefined. Let $d_u'$ be the last value encountered in the linear traversal
such that $i_u(d_u')$ is defined. Then $d_u'$ is the largest element in $\mathcal D_\epsilon$ that $\tau_{S_u}$ maps to
$\tau_{S_u}(d_u')$ and $d_u$ is larger than the smallest such element. Hence,
$\tau_{S_u}(d_u) = \tau_{S_u}(d_u')$ and it follows from Lemma~\ref{Lem:tau} that $i_u$ need
not be updated from the
value it had when $d_u'$ was encountered. Thus, maintaining $i_u$ is easy, assuming its initial value can be identified.

What if $i_u(d_u)$ is undefined for the initial $d_u$? Then we move down the list $\mathcal D_\epsilon$
until we find an index $i_u(d_u')$ that is defined. By Lemma~\ref{Lem:tau}, this index is the initial
value of $i_u$ and we are done. The problem with this approach is that we may need to traverse a large part
of the list before the index can be found. We can only afford to traverse $O(1/\log(1 + \epsilon))$
entries of $\mathcal D_\epsilon$. The following lemma shows that if the search has
not identified an index $i_u(d_u)$ after a small number of steps then our oracle can output twice the distance value
in the final entry considered.
\begin{Lem}\label{Lem:NoPointer}
For vertices $u$ and $v$, let $j$ be the index of $\mathcal D_\epsilon$ such that
$\mathcal D_\epsilon[j] = \tilde d_{uv}$.
Assume that $j_{\min} = j - \lceil\log(2\alpha_{MN})/\log(1 + \epsilon)\rceil$ is an index of $\mathcal D_\epsilon$ such that
$i_u(\mathcal D_\epsilon[j'])$ and $i_v(\mathcal D_\epsilon[j'])$ are undefined for all $j_{\min}\leq j'\leq j$.
Then $d_G(u,v)\leq 2\mathcal D_\epsilon[j_{\min}]\leq (1 + \epsilon)k\cdot d_G(u,v)$.
\end{Lem}
\begin{proof}
We have $d_G(u,v)\leq\mathcal D_\epsilon[j]\leq (1 + \epsilon)\alpha_{MN}k\cdot d_G(u,v)$.
For each index $j' > 0$ of $\mathcal D_\epsilon$, $\mathcal D_\epsilon[j'-1] = \mathcal D_\epsilon[j']/(1 + \epsilon)$ by
Lemma~\ref{Lem:EpsilonComb} and part $2$ of Lemma~\ref{Lem:Comb}. Thus,
\[
  \mathcal D_\epsilon[j_{\min}] = \frac{\mathcal D_\epsilon[j]}{(1 + \epsilon)^{j - j_{\min}}}
                             \leq \frac{(1 + \epsilon)\alpha_{MN}k}{(1 + \epsilon)^{\log(2\alpha_{MN})/\log(1 + \epsilon)}}d_G(u,v)
                             = \frac {(1 + \epsilon)k}2 d_G(u,v),
\]
showing the second inequality of the lemma.

To show the first inequality, let $I\in\mathcal I_u$ be the interval containing $\mathcal D_\epsilon[j]$. Then it follows from
Lemma~\ref{Lem:tau} that $\mathcal D_\epsilon[j']\in I$ for every $j'$ satisfying the condition in the lemma.
Recalling our assumption $\epsilon\leq \frac 1 2 < 1 - 1/\alpha_{MN}$, we get
$(1 + \epsilon)^{j - j_{\min}} \geq 2\alpha_{MN} > 2/(1 - \epsilon)$ so
\[
  \mathcal D_\epsilon[j] - \mathcal D_\epsilon[j_{\min}]
                           =    \left(1 - \frac 1{(1 + \epsilon)^{j - j_{\min}}}\right) \mathcal D_\epsilon[j]
                           >    \left(1 - \frac{1 - \epsilon}2\right)\mathcal D_\epsilon[j]
                           >    \frac 1 2 d_G(u,v)
\]
and since $\mathcal D_\epsilon[j], \mathcal D_\epsilon[j_{\min}]\in I$,
$I$ must have length $> \frac 1 2 d_G(u,v)$. Let $j_u$ be the index of $P_u$ such that
interval $I_u = [P_u[j_u],P_u[j_u + 1]]$ contains
$I$. Since $I$ is one of four consecutive subintervals of $I_u$ of even length, $I_u$ has length $> 2d_G(u,v)$.
Also, $P_u[j_u]\leq \mathcal D_\epsilon[j_{\min}]$.

Similarly, there is an index $j_v$ of $P_v$ such that
$I_v = [P_v[j_v],P_v[j_v + 1]]$ has length
$> 2d_G(u,v)$ and $P_v[j_v]\leq \mathcal D_\epsilon[j_{\min}]$.


Let $j$ be the final index of $\dist_k(u,v,0)$ (corresponding to a $uv$-query to the Thorup-Zwick oracle). Assume it is even
(the odd case is handled in a similar manner). Then
$d_G(u,p_{j' + 2}(u)) - d_G(u,p_{j'}(u))\leq 2d_G(u,v)$ for all even $j'\leq j - 2$ (using an observation similar to that in
the proof of Lemma~\ref{Lem:FeasibleRec}). By the above,
$P_u[j_u]\geq d_G(u,p_j(u))$. We also have
$d_G(v,p_{j' + 2}(v)) - d_G(v,p_{j'}(v))\leq 2d_G(u,v)$ for all odd $j'\leq j - 3$ so again by the above,
$P_v[j_v]\geq d_G(v,p_{j-1}(v))$. Finally, since $p_{j-1}(v)\notin B_u$,
\begin{align*}
  d_G(v,p_j(u)) & \leq d_G(u,v) + d_G(u,p_j(u))\\
                & \leq d_G(u,v) + d_G(u,p_{j-1}(v))\\
                & \leq 2d_G(u,v) + d_G(v, p_{j-1}(v)).
\end{align*}
Thus, $d_G(v,p_j(u)) - d_G(v,p_{j-1}(v))\leq 2d_G(u,v)$ and since $p_j(u)\in B_v$ we have $d_G(v,p_j(u))\in P_v$. Also,
$d_G(v,p_{j-1}(v))\in P_v$ so since $P_v[j_v]\geq d_G(v,p_{j-1}(v))$, we get $P_v[j_v]\geq d_G(v,p_j(u))$.
We can now conclude the proof with the first inequality of the lemma:
\[
  d_G(u,v) \leq d_G(u,p_j(u)) + d_G(v,p_j(u))
           \leq P_u[j_u] + P_v[j_v]
           \leq 2\mathcal D_\epsilon[j_{\min}].
\]
\end{proof}

\subsection{Running time and space}\label{subsec:TimeSpace}
We now bound the time and space of our oracle.

\paragraph{Preprocessing:} Constructing the Mendel-Naor oracle takes
$O(n^{2 + 1/k}\log n)$ time and requires $O(n^{1 + 1/k})$ space. Traversing the nodes of the trees kept by
the oracle identifies all distances in time proportional to their number which by Lemma~\ref{Lem:DMNBound} is $O(n^{1 + 1/k})$.
Sorting them to get the list representation of $\mathcal D_{MN}$ then takes $O(n^{1 + 1/k}\log n)$ time.

Forming a sorted list of the values from $\cup_{d\in\mathcal D_{MN}}\mathcal D_d$ in Lemma~\ref{Lem:EpsilonComb} can be done in
$O((|\mathcal D_{MN}|/\log(1 + \epsilon))\log n) = O(\frac 1 \epsilon n^{1 + 1/k}\log n)$ time and requires
$O(\frac 1 \epsilon n^{1 + 1/k})$ space.
Clearly, when the input to $\comb_\epsilon$ is given as a sorted list, the algorithm can be implemented to run in time linear in
the length of the list. Thus, computing a sorted list of the values of $\mathcal D_\epsilon$ can be done in
$O(\frac 1 \epsilon n^{1 + 1/k}\log n)$ time.

By the analysis of Thorup and Zwick, forming bunches $B_u$ takes $O(kmn^{1/k})$ time.
Since these bunches have total size
$O(kn^{1 + 1/k})$, sorted lists $P_u$ can be found in $O(kn^{1 + 1/k}\log n)$ time. Sets $S_u$ can be found within the
same time bound.

Forming $\mathcal D_\epsilon(u)$-sets can be done by two linear traversals of the sorted list $L$ of values from
$\mathcal D_\epsilon\cup\bigcup_{u\in V}S_u$. The first traversal visits elements in decreasing order. Whenever we encounter
a $d$ from a set $S_u$, let $d'$ be the previous visited element of $S_u$ ($d' = \infty$ if no such element exists)
and let $d''$ be the latest visited element of $\mathcal D_{\epsilon}$. If $d\leq d'' < d'$,
$d''$ is the smallest element of $\mathcal D_\epsilon$ that $\tau_{S_u}$ maps to
$\tau_{S_u}(d'') = d$ so we add $d''$ to $\mathcal D_\epsilon(u)$. Otherwise we do nothing as $\tau_{S_u}$ maps no element of
$\mathcal D_\epsilon$ to $d$. The second traversal visits elements in
increasing order. When we encounter a $d\in S_u$, let $d'$ be the predecessor of $d$ in $S_u$ ($d' = -\infty$ is no such element
exists) and let $d''$ be the latest visited
element of $\mathcal D_{\epsilon}$. Then, assuming $d'\leq d'' < d$, $d''$ is the largest element that $\tau_{S_u}$ maps
to $\tau_{S_u}(d'') = d'$ and so we add $d''$ to $\mathcal D_{\epsilon}(u)$. Together, these two traversals form
all $\mathcal D_\epsilon(u)$-sets in time $O(|\mathcal D_\epsilon| + \sum_{u\in V}|S_u|)$.

Since each element of each set $S_u$ is associated with at most two elements of $\mathcal D_{\epsilon}(u)$, we get a
space bound of $O(kn^{1 + 1/k})$ for sets $\mathcal D_\epsilon(u)$. In the two traversals, we can easily identify $i_u(d)$,
$d\in\mathcal D_\epsilon(u)$, without an asymptotic increase in time. We represent each of
these index maps as hash functions in the same way as bunches $B_u$ are represented in the Thorup-Zwick oracle. These hash functions
do not increase space.

\paragraph{Query:} To answer a $uv$-query, we need an efficient implementation of algorithm $\refine$. The while-loop consists of
$O(1/\epsilon)$ iterations. Sub-routine $\refinefurther$ can be implemented to run in constant time assuming we have
precomputed, for each $u$ and each even index $i_u\geq 2$, the even index $j$ in $0,\ldots,i_u - 2$ that maximizes $\delta_j$.
This preprocessing can easily be done in $O(kn)$ time. It then follows that $\refine$ runs in $O(1/\epsilon)$ time
and we can conclude with our second main result.
\begin{theorem}\label{Thm:ConstQuery}
For any integer $k\geq 1$ and any $0 < \epsilon\leq 1$, a $((2 + \epsilon)k)$-approximate distance oracle of $G$ of size
$O(kn^{1 + 1/k})$ and query time $O(1/\epsilon)$
can be constructed in $O(n^{2 + 1/k}\log n)$ time. For $k = O(\log n/\log\log n)$ and constant $\epsilon$, space can be
improved to $O(n^{1 + 1/k})$.
\end{theorem}
\begin{proof}
We may assume that $k\geq 4$ since otherwise we can apply the Thorup-Zwick oracle or our $O(\log k)$ query time oracle.
Apply Lemma~\ref{Lem:Generic} and Lemma~\ref{Lem:NoPointer} with $\epsilon' = \frac 1 2\epsilon\leq\frac 1 2$ instead of $\epsilon$.
Then we get stretch $(2 + \epsilon)k$, size $O(kn^{1 + 1/k})$, and query time $O(1/\epsilon)$.
This shows the first part of the theorem.

To show the second part, apply the first part with $\epsilon_1 = \frac 1 2 \epsilon$ instead of $\epsilon$ and $k' = k(1 + \epsilon_2)$
instead of $k$, where $\epsilon_2 = \epsilon/(4 + \epsilon)$ (we assume here for simplicity that $k(1 + \epsilon_2)$ is an integer).
Then $(2 + \epsilon_1)k' = (2 + \epsilon)k$ so we get the desired stretch. Size is $O(k'n^{1 + 1/k'}) = O(kn^{1 + 1/k'})$.
Letting $\epsilon_3 = \epsilon_2/(1 + \epsilon_2)$, we have $1/k' = (1 - \epsilon_3)/k$ so we get size $O(n^{1 + 1/k})$ if
$kn^{-\epsilon_3/k}\leq 1$, i.e., if $k\log k\leq \epsilon_3\log n$. The latter holds when $k = O(\log n/\log\log n)$.
\end{proof}

\subsection{Faster preprocessing}\label{subsec:FasterPreproc}
In this subsection, we show how to improve the $O(n^{2 + 1/k}\log n)$ preprocessing bound in
Theorem~\ref{Thm:ConstQuery}. First, we can replace the Mendel-Naor oracle with that of Mendel and Schwob~\cite{CKR}.
This follows since the latter also uses ultrametric representations of approximate shortest path distances so the proof of
Lemma~\ref{Lem:DMNBound} still holds. This modification alone gives a preprocessing bound of $O(mn^{1/k}\log^3n)$.

Next, observe that our result holds for any $O(k)$-approximate distance $d_{MN}(u,v)$ output
and not just for $\alpha_{MN} = 128$. More
precisely, let $C > 1$ be an integer. If $d_{MN}(u,v)$ has stretch $Ck$ then it follows from our analysis that
this estimate can be refined to $(2 + \epsilon)k$ in $O(\log C/\epsilon)$ iterations and we get preprocessing
time $O(mn^{1/(Ck)}\log^3n)$ and query time $O(\log C/\epsilon)$.
In addition to this, we need to construct bunches and form sorted lists $P_u$. As shown earlier, this can be done in
$O(kmn^{1/k} + kn^{1 + 1/k}\log n)$ time. Combining this with the above gives the following improvement in
preprocessing over that in Theorem~\ref{Thm:ConstQuery}.
\begin{theorem}
For any integers $k\geq 3$ and $C\geq 2$ and any $0 < \epsilon\leq 1$, a $((2 + \epsilon)k)$-approximate
distance oracle of $G$ of size $O(kn^{1 + 1/k})$ and query time $O(\log C/\epsilon)$
can be constructed in $O(kmn^{1/k} + kn^{1 + 1/k}\log n + mn^{1/(Ck)}\log^3n)$
time. For $k = O(\log n/\log\log n)$ and constant $\epsilon$, space can be improved to $O(n^{1 + 1/k})$.
\end{theorem}

\section{Concluding Remarks}\label{sec:ConclRem}
We gave a size $O(kn^{1 + 1/k})$ oracle with $O(\log k)$ query time for stretch $(2k - 1)$-distances, improving the $O(k)$ query time
of Thorup and Zwick. Furthermore, for any positive $\epsilon\leq 1$, we gave an oracle with stretch $(2 + \epsilon)k$ which answers
distance queries in $O(1/\epsilon)$ time. This improves the result of Mendel and Naor which answers stretch
$128k$-distances in $O(1)$ time.

For the first oracle, can we go beyond the $O(\log k)$ query bound? And can space be improved to $O(n^{1 + 1/k})$?
For the second oracle, can stretch be improved to $2k - 1$ while
keeping $O(1)$ query time? To our knowledge, the oracle of Mendel and Naor cannot be used to produce approximate shortest
paths, only distances. Our second oracle then has the same drawback (due to Lemma~\ref{Lem:NoPointer}). What can be done to deal
with this?

\section*{Acknowledgments}
I thank anonymous referees for their comments and remarks that helped improve
the presentation of the paper.

\end{document}